\documentclass[a4paper]{amsproc}
\usepackage{amssymb}
\usepackage{amscd}
\usepackage{amsfonts}
\usepackage{amssymb}
\usepackage{cite}
\usepackage{amsmath}
\usepackage{epsfig}
\usepackage{graphicx, subfigure}
\usepackage{color, graphics}
\usepackage{pdfsync}
\usepackage{hyperref}
\usepackage[all,cmtip]{xy}
\usepackage{setspace}

\theoremstyle{plain}
 \newtheorem{thm}{Theorem}[section]
 \newtheorem{prop}{Proposition}[section]

\theoremstyle{definition}
 
 \newtheorem{rem}{Remark}[section]

\numberwithin{equation}{section}

\def\ji {\char'032}

\def\m  {\char'176}

\font\srit=wncyi8
 \font\srrm=wncyr8

\newcommand{\R}{\mathbb{R}}

\allowdisplaybreaks

\title[D'Alamebrt principle and classical relativity in Lagrangian mechanics]
{D'Alamebrt and Hamiltonian principles and classical relativity in Lagrangian mechanics}
\author[Jovanovi\'c]{
Bo\v zidar Jovanovi\'c}

\address{
Mathematical Institute SANU \\
Serbian Academy of Sciences and Arts \\
Kneza Mihaila 36, 11000 Belgrade\\
Serbia}
\email{bozaj@mi.sanu.ac.rs}

\begin{document}

\begin{abstract}
In this note we present invariant formulation of the d'Alambert principle and classical time-dependent Lagrangian mechanics with holonomic constraints from the perspective of moving frames.
\end{abstract}

\maketitle


\section{Introduction}

We consider a classical mechanical Lagrangian system $(Q,L)$, where $Q$ is an $n$--dimen\-sional configuration space and $L\colon TQ\times\R\to\R$ is a time-dependent Lagrangian.

The d'Alambert principle or d'Alambert-Lagrange principle states that the trajectories of a mechanical system can be
obtained from the condition that the variational derivative of the Lagrangian vanishes along virtual displacements \cite{Ap, Ar, AKN}.
It is one of the basic tools in mechanics and is applied to both holonomic and nonholonomic systems. For holonomic systems, the principle is equivalent to the Hamiltonian principle of least action.
Motivated by the notion of fixed and moving reference frames in rigid body dynamics \cite{Ar, AKN}, we consider arbitrary time-dependent transformations between the configuration space $Q$ (the fixed reference frame) and the manifold $M$ diffeomorphic to $Q$ (the moving reference frame)
and consider trajectories of a classical Lagrangian system in both reference frames (section \ref{s2}). In particular, we consider systems with time-dependent holonomic constraints (section \ref{s3}). All considerations are valid without the assumption that the Lagrangian is regular and are derived without the use of Lagrange multipliers. For this reason, we do not discuss the uniqueness of the solution.

In section \ref{s4} we apply the construction of moving reference frames for the invariant formulation of classical Lagrangian mechanics in a space-time, $(n+1)$--dimensional manifold $\mathcal Q$ fibred over $\R$ with fibers diffeomorphic to $Q$.
The invariant formulation of time-dependent classical Lagrangian mechanics is well studied (see e.g. \cite{MP, GM} and references therein). Here we have tried to present it with minimal technical requirements.

All considered objects in the note are assumed to be smooth.

\section{D'Alambert principle and dynamics of relative motions}\label{s2}

\subsection{D'Alambert principle and Hamiltonian principle of least action}

Let us consider a Lagrangian system $(Q,L)$, where $Q$ is an $n$--dimensional configuration space and $L(q,\dot q,t)$ is a Lagrangian, $L: TQ \times \R\to
\R$. Let $q=(q^1,\dots,q^n)$ be local coordinates on $Q$. A curve $\gamma\colon q(t)=(q^1(t),\dots,q^n(t))$ is a
motion of the system if it satisfies the \emph{Euler--Lagrange equations}
\begin{equation} \label{EL}
\frac{d}{dt}\frac{\partial L}{\partial \dot q^i}=\frac{\partial
L}{\partial q^i}, \qquad i=1,\dots,n.
\end{equation}

The solutions of the Euler--Lagrange equations are exactly the critical points of the action integral
\begin{equation}\label{action}
S_L({\gamma})=\int_a^b L(\gamma,\dot\gamma,t)dt
\end{equation}
in a class of curves
\begin{equation}\label{krive}
\varOmega=\varOmega(Q,q_0,q_1,a,b)=\{
{\gamma}\colon [a,b]\to Q, \, {\gamma}(a)=q_0, {\gamma}(b)=q_1\}
\end{equation}
with fixed endpoints (the \emph{Hamiltonian principle of least action}
(1834), see e.g. \cite{Ar, AKN, BKKT, Po}).
Namely, let $\gamma\in \varOmega$, let $\gamma_s\in\varOmega$, $s\in(-\epsilon,\epsilon)$ be a smooth variation of $\gamma$ ($\gamma_0=\gamma$), and let
\[
\eta=\frac{d}{ds}\vert_{s=0}(\gamma_s)
\]
be the corresponding vector field along the curve $\gamma\subset Q$. Then
\begin{equation}\label{izvod}
\lim_{s\to 0}\frac{1}{s}\big(S[\gamma_s]-S[\gamma]\big)=\int_a^b \delta L(\eta)\vert_\gamma dt,
\end{equation}
where $\delta L(\eta)\vert_\gamma$ is a variational derivative of $L$ along $\gamma$ in a direction of $\eta$.
Let $\xi$ and $\gamma$ be arbitrary time-dependent vector fields and a curve on $Q$.
The \emph{variational derivative of $L$ along $\gamma$ in the direction of $\xi$} is defined by
\begin{equation}\label{varijacioniIzvod}
\delta L(\xi)\vert_\gamma=\sum_{i=1}^n \big( \frac{\partial L}{\partial q^i}-\frac{d}{dt}\frac{\partial L}{\partial \dot q^i}\big)\xi^i\vert_\gamma,
\end{equation}
where $\xi=\sum_i\xi^i(q,t){\partial/\partial q^i}$ and $\gamma$ is given by $q(t)=(q^1(t),\dots,q^n(t))$ in a local coordinate system $q=(q^1,\dots,q^n)$.
Therefore, $\gamma(t)$ is an extremal of the action functional if and only if it satisfies the Euler--Lagrange equations \eqref{EL}.

In classical mechanics, for a given curve $\gamma$, $q(t)=(q^1(t),\dots,q^n(t))$, a time-dependent vector field $\eta$ along $\gamma$ is usually referred to as a vector field of \emph{virtual displacements}. By using the variational derivative, we can formulate the dynamics in terms of the \emph{d'Alambert principle}:
a curve $\gamma(t)$ is a motion of the Lagrangian system $(Q,L)$ if the variational derivative $\delta L(\eta)$ is equal to zero,
\begin{equation*}\label{Dalamber1}
\delta L(\eta)\vert_\gamma=\sum_{i=1}^n \big( \frac{\partial L}{\partial q^i}-\frac{d}{dt}\frac{\partial L}{\partial \dot q^i}\big)\eta^i\vert_\gamma=0,
\end{equation*}
for all virtual displacements $\eta$ along $\gamma$ \cite{Ap, Ar, AKN, BKKT}.

Although the statement of the principle is equivalent to the fact that an element of a dual space of a vector space is zero if and only if its kernel is the entire vector space, it is fundamental for the formulation of Lagrangian mechanical systems with constraints.

\subsection{Dynamics in the moving frames}

In analogy to rigid body dynamics, where we consider the fixed and the moving frames defined by time-dependent isometries of Euclidean space, we consider \emph{the moving reference frame} in a general Lagrangian system $(Q,L)$ as a time-dependent diffeomorphism
\footnote{On the other hand, note that in \cite{FA}  the rigid body dynamics is considered from the perspective of continuum mechanics.}
\[
g_t\colon M\to Q, \qquad q=g_t(x).
\]
Here $M=Q$, but we use different symbols to emphasise the domain and codomain of the mapping: the variable $q$ is in the fixed frame, while the variable
$x$ is in the moving reference frame.
Furthermore, in analogy to the angular velocity in the fixed and in the moving frame, we define the time-dependent vector fields $\omega_t\in \mathfrak X(Q)$ and $\Omega_t\in \mathfrak X(M)$
by the identities
\[
\omega_t(q)=\frac{d}{ds}\vert_{s=0}\big(g_{t+s}\circ (g_t)^{-1}(q)), \qquad \Omega_t(x)=\frac{d}{ds}\vert_{s=0}\big((g_{t})^{-1}\circ (g_{t+s})(x)).
\]

Note that $g_t$ is a curve in the Lie group $\text{\sc Diff}(Q)$ and the vector fields $\omega_t$ and $\Omega_t$ are elements of its Lie algebra
$\mathfrak X(Q)=Lie(\text{\sc Diff}(Q))$\footnote{Here we are not interested in the smooth structures of these infinite-dimensional objects} given by the right and left translation of the velocity $\dot g_t$. They are related by the adjoint mapping: $\omega_t=\mathrm{Ad}_{g_t}(\Omega_t)$. Furthermore, as in the case of the rigid body we have the following proposition.

\begin{prop}
(i) The angular velocity vector fields are related as follows
\begin{equation*}\label{ugaone}
dg_t(\Omega_t)\vert_x=\omega_t\vert_{q=g_t(x)}.
\end{equation*}

(ii) {\sc The addition of velocities} Let $\Gamma(t)$ be a smooth curve on $M$ and $\gamma(t)=g_t(\Gamma(t))$ be the associated curve on $Q$. Then
\[
\dot\gamma=dg_t(\dot\Gamma))+\omega_t(g_t(\Gamma(t))).
\]
Conversely, for a given curve $\gamma(t)$ and the associated curve $\Gamma(t)=g_t^{-1}(\gamma(t))$,
we have
\[
\dot\Gamma=dg_t^{-1}(\dot\gamma))-\Omega_t(g_t^{-1}(\gamma(t))).
\]
\end{prop}

For a given Lagrangian $L\colon TQ\times R\to \R$, we define the associated Lagrangian $l$ in the moving frame by
\begin{equation}\label{noviLagranzijan}
l(x,\dot x,t):=L(q,\dot q,t)\vert_{q=g_t(x),\dot q=dg_t(\dot x)+\omega_t(g_t(x))}.
\end{equation}

The following observation is of fundamental importance for the further explanations.

\begin{prop}\label{glavna}
Let $\Gamma(t)$ and $\xi$ be a smooth curve and a vector field in the moving frame and $\gamma(t)=g_t(\Gamma(t))$ and $\eta=dg_t(\xi)$ the associated curve and the vector field in the fixed frame, $t\in\R$.
Then the variational derivative of $L$ along $\gamma$ in the direction of $\xi$ coincides with the variational derivative of $l$ along $\Gamma$
in the direction of $\eta$:
\begin{equation}\label{invDal}
\delta L(\eta)\vert_{\gamma}=\delta l(\xi)\vert_{\Gamma}.
\end{equation}
\end{prop}

\begin{proof}
Consider local coordinates $q=(q^1,\dots,q^n)$ and $x=(x^1,\dots,x^n)$ defined in the domains $U$ and $V$, neighborhoods of the points $q_0=g_{t_0}(x_0)\in Q$ and $x_0\in M$. The mapping $g_t$ is given by functions
\begin{equation}\label{preslikavanje1}
q^i=Q^i(x,t)=Q^i(x^1,\dots,x^n,t), \qquad i=1,\dots,n,
\end{equation}
defined for some time interval $I$ ($t_0\in I$),
\[
g_t(V)\subset U, \qquad t\in I.
\]
Then the angular velocity vector field in the fixed frame is given by
\begin{equation*}
\omega_t(g_t(x))=\sum_{i=1}^n \frac{\partial Q^i(x,t)}{\partial t}\frac{\partial}{\partial q^i}.
\end{equation*}

The curves $\Gamma(t)$ and $\gamma(t)=g_t(\Gamma(t))$, $t\in I$, are given locally by the functions $x^i=x^i(t)$ and $q^i=Q^i(x(t),t)$.
The law of addition of the velocities is
\begin{equation}\label{preslikavanje2}
\dot q^i=\sum_{j=1}^n \frac{\partial Q^i(x,t)}{\partial x^j}\dot x^j+\frac{\partial Q^i(x,t)}{\partial t}.
\end{equation}

By plugging \eqref{preslikavanje1} and \eqref{preslikavanje2} into \eqref{noviLagranzijan}, after straightforward computations,
we get that along $\Gamma(t)$ we have (see e.g. \cite{Zu})
\begin{equation*}\label{forme}
\frac{\partial l}{\partial x^i}-\frac{d}{dt}\frac{\partial l}{\partial \dot x^i}=\sum_{j=1}^n
\big(\frac{\partial L}{\partial q^j}-\frac{d}{dt}\frac{\partial L}{\partial \dot q^j}\big)\frac{\partial Q^j}{\partial x^i},\qquad i=1,\dots,n,
\end{equation*}
which together with the definitions of $\eta$,
\[
\eta^j\vert_{q=g_t(x)}=\sum_{k=1}^n \frac{\partial Q^j}{\partial x^k}\xi^k\vert_x, \qquad j=1,\dots,n,
\]
and the variational derivative \eqref{varijacioniIzvod}, imply the statement of the proposition.
\end{proof}

\begin{rem}
If we consider the functions $Q^i$ that do not depend on time ($\omega_t\equiv 0$), the proof of Proposition \ref{glavna} is actually the proof that the variational derivative \eqref{varijacioniIzvod} does not depend on the coordinate system of the configuration space $Q$.
\end{rem}

By setting both sides of the equality \eqref{invDal} to zero,
Proposition \ref{glavna} can be interpreted as:
\emph{If a motion satisfies the  d'Alambert principle in one reference frame, then it satisfies the principle in arbitrary reference frame}

As a direct consequence, we also obtain the invariance of the Hamitonian principle of least action.

\begin{thm} \label{posledica1}
Let $x_0=g_a^{-1}(q_0)$ and $x_1=g_b^{-1}(q_1)$.
A curve $\gamma(t)$ in fixed space $Q$ is the critical point of the action integral \eqref{action} in the class of curves \eqref{krive} if and only if the curve $\Gamma(t)=g_t^{-1}(\gamma(t))$ in the moving frame $M$
the critical point of the action integral
\begin{equation*}\label{action*}
S_l({\Gamma})=\int_a^b l(\Gamma,\dot\Gamma,t)dt
\end{equation*}
in a class of curves
\begin{equation*}\label{krive*}
\varOmega(M,x_0,x_1,a,b)=\{
{\Gamma}\colon [a,b]\to M, \, {\Gamma}(a)=x_0, {\Gamma}(b)=x_1\}.
\end{equation*}
\end{thm}

\begin{rem}
Let us recall on the Galilean principle of relativity (see \cite{Ar, P1}):

\begin{itemize}

\item \emph{All laws of nature are the same at all times in all inertial coordinate systems}.

\item \emph{A coordinate system that is in uniform rectilinear motion with respect to an inertial frame is also inertial}.

\end{itemize}

Depending on the geometric structure of the 4-dimensional affine space-time, we obtain two different mechanics: classical and special relativity, which have the same principle of relativity (see e.g. \cite{Jo}).
On the other hand, d'Alambert's principle satisfies the following general variant of the principle of relativity, which does not include any notion of inertial frames:

\begin{itemize}

\item \emph{All laws of nature are the same at all times in all reference systems}

\end{itemize}

Again, with suitable geometric structures on space-time manifolds, we can obtain both the general theory of relativity and classical mechanics. 
The invariant formulation of classical Lagrangian mechanics on a space-time manifold is well known
(see e.g. \cite{MP, GM} and references therein). Here, in section \ref{s4}, we have tried to present it with minimal technical requirements.
\end{rem}

\section{D'Alambert principle for systems with holonomic constraints}\label{s3}

\subsection{Time-dependent holonomic constraints}

We now consider the
Lagrangian system $(Q,L,\Upsilon_t)$ in which a motion $\gamma(t)$ is restricted to time-dependent $m$--dimensional \emph{immersed} submanifolds $\Upsilon_t=g_t(\Sigma)$,
\[
g_t\colon \Sigma\to \Upsilon_t \subset Q, \qquad t\in\R,
\]
without selfintersections.
The curves $\gamma\colon \R\to Q$ that satisfy $\gamma(t)\in\Upsilon_t$ ($t\in\R$) are called \emph{admissible}.

Typically, we have a situation like in the previous section, where $\Sigma$ is a fixed $m$--dimensional immersed submanifold of $M=Q$ without selfintersections.
In mechanics, the above restrictions on a motion of the system are usually called time-dependent (or rheonomic) \emph{holonomic constraints}.

\begin{rem}
Usually one considers \emph{embedded} submanifolds $\Upsilon_t=g_t(\Sigma)$, $t\in\R$. However, we can assume a non-holonomic system on $Q$ with integrable constraints.
As a result, $Q$ is foliated on integral immersed submanifolds without selfintersections and we can consider a motion on a single integral leaf. For example, let us consider a system describing a rolling without sliding of a disk of radius $r$ over a disk of radius $R$. The configuration space is a torus $\mathbb T^2$. It is known that the corresponding nonholonomic constraint is integrable. If $r/R\notin\mathbb Q$, the configuration space is foliated on everywhere dense integral curves of the distribution.
\end{rem}

Consider local coordinates $q=(q^1,\dots,q^n)$ and $x=(x^1,\dots,x^m)$, defined in the domains $U$ and $V$, neighborhoods of the points $q_0=g_{t_0}(x_0)\in Q$ and $x_0\in \Sigma$.
The mapping $g_t$ is given locally by functions
\begin{equation}\label{preslikavanje3}
q^i=Q^i(x,t)=Q^i(x^1,\dots,x^m,t), \qquad i=1,\dots,n,
\end{equation}
defined for some time interval $I$ ($t_0\in I$),
\[
g_t(V)\subset U\subset Q, \qquad t\in I.
\]

Locally, the submanifolds $g_t(V)\subset \Upsilon_t$ can be given by equations
\begin{equation}\label{veze}
f_\alpha(q^1,\dots,q^n,t)=0, \qquad \alpha=1,\dots,n-m,
\end{equation}
where the rank of the matrix $\partial f_\alpha/\partial q^j$ is equal to $n-m$ and
\[
f_\alpha(Q^1(x^1,\dots,x^m,t),\dots, Q^n(x^1,\dots,x^m,t), t) \equiv 0, \quad x \in V\subset \Sigma, \quad t\in I.
\]

Now we can only define the time-dependent angular velocity vector field $\omega_t$ as a section of $T_{\Upsilon_t}Q$:
\[
\omega_t(q)=\frac{d}{ds}\vert_{s=0}\big(g_{t+s}\circ (g_t)^{-1}(q)), \quad q=g_t(x), \quad x\in\Sigma,
\]
i.e,
\[
\omega_t(g_t(x))=\sum_{i=1}^n \frac{\partial Q^i(x,t)}{\partial t}\frac{\partial}{\partial q^i}, \qquad x\in \Sigma.
\]

Let $\Gamma$ be a smooth curve on $\Sigma$ and let $\gamma(t)=g_t(\Gamma(t))$ be the associated curve on $Q$. Then, as in the case of moving frames considered above, we have
\[
\dot\gamma=dg_t(\dot\Gamma)+\omega_t(g_t(\Gamma(t))),
\]
or locally,
\begin{equation}\label{preslikavanje4}
\dot q^i=\sum_{j=1}^m \frac{\partial Q^i(x,t)}{\partial x^j}\dot x^j+\frac{\partial Q^i(x,t)}{\partial t}.
\end{equation}

In other words, the \emph{admissible velocities} (velocities allowed by holonomic constraints) at the point $q=g_t(x)\in\Upsilon_t$ belong to the affine subspace of the tangent bundle $T_{\Upsilon_t} Q$:
\[
\mathcal A_{q,t}=\omega_t(g_t(\Gamma(t)))+dg_t(T_x\Sigma))=\omega_t(g_t(\Gamma(t)))+T_q \Upsilon_t \subset T_q Q.
\]

The vectors that are tangent to the constraint submanifolds $\eta\in T_q \Upsilon_t\subset T_q Q$ and $\xi\in T_x\Sigma\subset T_x M$ are called \emph{virtual displacements}:
\[
\mathcal V_{q,t}=T_q \Upsilon_t=dg_t(\mathcal V_{x}), \qquad \mathcal V_{x}=T_x\Sigma.
\]
Note that in the ``moving frame'' $\Sigma\subset M$ the space of admissible velocities coincides with the space of virtual displacements: $\mathcal A_x=\mathcal V_x=T_x\Sigma$.

By using the constraint \eqref{veze}, the spaces of admissible velocities and virtual displacements are usually described by the equations
\[
\frac{\partial f_\alpha}{\partial t}+\sum_{i=1}^n \frac{\partial f_\alpha}{\partial q^i}\xi^i=0, \quad \text{and} \quad
\sum_{i=1}^n \frac{\partial f_\alpha}{\partial q^i}\xi^i=0,
\quad \alpha=1,\dots,n-m,
\]
respectively.

\subsection{External description of the dynamics: d'Alembert principle}

In classical mechanics, dynamics in the case of \emph{ideal holonomic constraints} is defined by the \emph{d'Alembert principle}:
a curve $\gamma(t)\in\Upsilon_t$ is a motion of the constrained Lagrangian system $(Q,L,\Upsilon_t)$ if the variational derivative $\delta L(\eta)\vert_\gamma$ vanishes for all virtual displacements $\eta$ along $\gamma$:
\begin{equation}\label{dalamber}
\delta L(\eta)\vert_\gamma=\sum_{i=1}^n \big(\frac{\partial L}{\partial q^i}-\frac{d}{dt}\frac{\partial L}{\partial \dot q^i}\big)\eta^i\vert_{\gamma}=0, \qquad \eta\vert_{\gamma(t)}\in \mathcal V_{\gamma(t)}=T_{\gamma(t)}\Upsilon_t.
\end{equation}

In the case of ideal holonomic constraints, the d'Alambert principle is equivalent to the Hamiltonian principle of least action:
the constrained motions are critical points of the action integral \eqref{action}
in a class of curves
\begin{equation}\label{krive2}
\varOmega=\varOmega(\Upsilon_t,q_0,q_1,a,b)=\{
{\gamma}\colon [a,b]\to Q, \, \gamma(t)\in\Upsilon_t,\, {\gamma}(a)=q_0, {\gamma}(b)=q_1\}.
\end{equation}
Namely, let $\gamma_s\in\varOmega$, $s\in(-\epsilon,\epsilon)$ be a smooth variation of $\gamma\in\varOmega$, $\gamma_0=\gamma$.
Then
\[
\eta(t)=\frac{d}{ds}\vert_{s=0}(\gamma_s)\in \mathcal V_{\gamma(t)}=T_{\gamma(t)}\Upsilon_t,
\]
and the statements follow from the identity \eqref{izvod}.

On the other hand the Hamiltonian principle of least action is not equivalent to the d'Alambert principle in the case of non-holonomic cost constraints.
One of the attempts to formulate non-holonomic dynamics from the variational principle is given by Kozlov (see e.g. \cite{AKN}).
We will discuss the non-holonomic constraints in a separate paper.

\subsection{Intrinsical formulation of the dynamics}

As in the cases of moving frames, for a given Lagrangian $L\colon TQ\times R\to \R$ we define the Lagrangian $l_\Sigma$ by
\begin{equation}\label{noviLagranzijan2}
l_\Sigma(x,\dot x,t):=L(q,\dot q,t)\vert_{q=g_t(x),\dot q=dg_t(\dot x)+\omega_t(g_t(x))}.
\end{equation}

Instead of solving the system \eqref{dalamber} in redundant variables $(q^1,\dots,q^n)$ using the method of Lagrange multipliers, it is more convenient to consider a standard Lagrangian system $(\Sigma,l_\Sigma)$. We can reduce the problem using the following analogue of Proposition \ref{glavna}.

\begin{prop}\label{glavna2}
Let $\Gamma(t)$ and $\xi$ be a smooth curve and a vector field of virtual displacements along $\Gamma$ in $\Sigma$ and let $\gamma(t)=g_t(\Gamma(t))$ and $\eta=dg_t(\xi)$ be the associated curve and the vector field of virtual displacements in $Q$. Then
\begin{equation}\label{jednakost}
\delta L(\eta)\vert_\gamma=\delta l(\xi)\vert_\Gamma.
\end{equation}
\end{prop}

\begin{proof}
The proof is the same as in the case of Proposition \ref{glavna}.
By setting \eqref{preslikavanje3} and \eqref{preslikavanje4} into \eqref{noviLagranzijan2}, after straightforward computations,  along $\Gamma(t)$ we obtain
\begin{equation*}\label{forme2}
\frac{\partial l}{\partial x^i}-\frac{d}{dt}\frac{\partial l}{\partial \dot x^i}=\sum_{j=1}^n
\big(\frac{\partial L}{\partial q^j}-\frac{d}{dt}\frac{\partial L}{\partial \dot q^j}\big)\frac{\partial Q^j}{\partial x^i}, \qquad i=1,\dots,m,
\end{equation*}
which together with the definitions of $\eta$,
\[
\eta^j\vert_{q=g_t(x)}=\sum_{k=1}^m \frac{\partial Q^j}{\partial x^k}\xi^k\vert_x, \qquad j=1,\dots,n,
\]
and the variational derivative \eqref{varijacioniIzvod}, imply \eqref{jednakost}.
\end{proof}

Since $T_{g_t(x)}\Upsilon_t=dg_t(T_x\Sigma)$, we get.

\begin{thm}\label{posledica2}
A curve $\gamma(t)\in\Upsilon_t$ is a motion of the constrained Lagrangian system $(Q,L,\Upsilon_t)$ if and only if $\Gamma(t)=g_t^{-1}(\gamma(t))$ is a motion of the Lagrangian system $(\Sigma,l_\Sigma)$. This means that in local coordinates $(x^1,\dots,x^m)$ on $\Sigma$ a curve $\Gamma\colon x(t)=(x^1(t),\dots,x^m(t))$ is a solution of the Euler--Lagrange equations
\begin{equation*} \label{EL2}
\frac{d}{dt}\frac{\partial l_\Sigma}{\partial \dot x^i}=\frac{\partial l_\Sigma}{\partial x^i}, \qquad i=1,\dots,m.
\end{equation*}
\end{thm}

Furthermore, for systems with constraints we also have

\begin{thm} \label{posledica3}
Let $x_0=g_a^{-1}(q_0)$ and $x_1=g_b^{-1}(q_1)$.
A curve $\gamma(t)$ in fixed space $Q$ is the critical point of the action integral \eqref{action} in a class of curves satisfying the constraints \eqref{krive2} if and only if the curve $\Gamma(t)=g_t^{-1}(\gamma(t))$
is the critical point of the action integral
\begin{equation*}\label{action**}
S_{l_\Sigma}({\Gamma})=\int_a^b l_\Sigma(\Gamma,\dot\Gamma,t)dt
\end{equation*}
in a class of curves
$\varOmega(\Sigma,\mathbf x_0,\mathbf x_1,a,b)=\{
{\Gamma}\colon [a,b]\to \Sigma, \, {\Gamma}(a)=\mathbf x_0, {\Gamma}(b)=\mathbf x_1\}$.
\end{thm}

In particular, we can consider the time-independent constraints,
\[
g_t \equiv \imath,
\]
where $\imath\colon \Sigma \hookrightarrow Q$ is the inclusion of the submanifold $\Sigma$ in $Q$.
Then
\[
l_\Sigma=L\vert_{T\Sigma}
\]
and Theorem \ref{posledica2} becomes the standard statement for systems with holonomic constraints that do not depend on time \cite{Ar, AKN}.

\begin{rem}
All considerations are valid without the assumption that the Lagrangian is regular and are derived without the use of Lagrange multipliers. Therefore, the uniqueness of the solution is not considered (see \cite{GM}). In the case where the Lagrangian is regular, we can apply the Legendre transformation and proceed to the Hamiltonian description of the time-dependent mechanical systems. In this case, we have studied Noether symmetries and integrability in contact and cosymplectic frameworks \cite{Jo2, Jo3}.
\end{rem}

\section{Space-time formulation of Lagrangian mechanics}\label{s4}

\subsection{Space-time and reference frames}

A space-time manifold in classical Lagrangian mechanics is an $(n+1)$--dimensional fiber manifold over real numbers
\begin{equation}\label{fibracija}
\tau\colon \mathcal Q \longrightarrow \R,
\end{equation}
where the fibers are diffeomorphic to an $n$--dimensional configuration space $Q$.

The points $\mathbf q$ in $\mathcal Q$ are called \emph{events} and the fibers $\tau^{-1}(a)$, $a\in\R$, are called spaces of \emph{simultaneous events}.
We say that the event $\mathbf q_0$ occurred before the event $\mathbf q_1$ if $\tau(\mathbf q_0)<\tau(\mathbf q_1)$.
A \emph{time line} (or \emph{world line})
is a smooth curve $s(t)$, a section of the fibration \eqref{fibracija}, $\tau(s(t))=t$.
 A time line $s(t)$, $t\in[a,b]$ is between (or connects) the events $\mathbf q_0$ and $\mathbf q_1$ if $s(a)=\mathbf q_0$
and $s(b)=\mathbf q_1$.

The space of \emph{virtual displacements} is a subbundle of $T\mathcal Q$, the vertical distribution of the fibration \eqref{fibracija}, defined by
\[
\mathcal V=\cup_{\mathbf q\in\mathcal Q}\mathcal V_{\mathbf q}, \qquad \mathcal V_{\mathbf q}=\ker d\tau\vert_{\mathbf q}=
\{\underline\xi\in T_{\mathbf q} \mathcal Q, d\tau\vert_{\mathbf q}(\underline \xi)=0\}.
\]

Since for time lines we have $d\tau(\dot{s}(t))=1$,
we also consider the affine subbundle of $T\mathcal Q$ (the first jet bundle \cite{MP}),
\[
\mathcal J=\cup_{\mathbf q\in\mathcal Q}\mathcal J_{\mathbf q}, \qquad
\mathcal J_{\mathbf q}=\{\underline\xi\in T_{\mathbf q} \mathcal Q, d\tau\vert_{\mathbf q}(\underline\xi)=1\}.
\]
It is clear that $\mathcal J$ is diffeomorphic to $\mathcal V$.

The (global) \emph{reference frame} is a trivialization
\begin{align*}
&\varphi_\alpha\colon \mathcal Q \longrightarrow Q_\alpha\times \R, \qquad Q_\alpha \cong Q,\\
& \varphi_\alpha(\mathbf q)=(q_\alpha,t_\alpha),
\end{align*}
such that
\[
\tau(\varphi^{-1}_\alpha(q_\alpha,t_\alpha))=t_\alpha+c_\alpha, \qquad c_\alpha\in \R.
\]
In other words: In the reference frame $\varphi_\alpha$ we set the time $t_\alpha$ to zero at the space of simultaneous events $\tau^{-1}(c_\alpha)$.

The vertical space $\mathcal V$ at $\mathbf q$ in the frame $\varphi_\alpha$ can be naturally identified with $T_q Q_\alpha\times \R$
\[
\underline \eta\in \mathcal V_{\mathbf q} \longleftrightarrow \eta_\alpha\in T_{q_\alpha} Q_\alpha \times\{t_\alpha\}, \quad (\eta_\alpha,0)=d\varphi_\alpha\vert_{\mathbf q}(\underline\eta), \quad (q_\alpha,t_\alpha)=\varphi_\alpha(\mathbf q).
\]

If we have two reference frames $\varphi_\alpha$ and $\varphi_\beta$, the transition function is defined by
\[
\phi_{\alpha\beta}=\varphi_\alpha\circ\varphi_\beta^{-1}\colon Q_\beta\times \R \longrightarrow Q_\beta\times \R
\]
is of the form
$
(q_\alpha,t_\alpha)=\phi_{\alpha\beta}(q_\beta,t_\beta)=\big(g_{\alpha\beta}(q_\beta,t_\beta),t_\beta+(c_\beta-c_\alpha)\big).
$

\subsection{Dynamics}
The Lagrangian $\mathbf L$ is a smooth function
\[
\mathbf L\colon\mathcal J\longrightarrow\R.
\]

Let $\gamma_\alpha(t_\alpha)$ be a curve in $Q_\alpha$. To $\gamma_\alpha$ we associate the time curve
\[
s(t)=\varphi_\alpha^{-1}((\gamma_\alpha(t_\alpha),t_\alpha))\vert_{t=t_\alpha+c_\alpha},
\]
and vice versa.
Since $\mathcal V$ and $\mathcal J$ are diffeomorphic, $\mathcal J$ can be identified with $TQ_\alpha\times \R$. In the context of this identification, the Lagrangian $\mathbf L$ in the reference frame $\varphi_\alpha$ is given by
\begin{align*}
& L_\alpha\colon TQ_\alpha\times \R \longrightarrow \R,\\
& L_\alpha(q_\alpha,\dot q_\alpha,t_\alpha)\vert_{q_\alpha=\gamma_\alpha(t_\alpha)}:=\mathbf L(\mathbf q,\dot{\mathbf q})\vert_{\mathbf q=s(t)=\varphi_\alpha^{-1}((\gamma_\alpha(t_\alpha),t_\alpha))\vert_{t=t_\alpha+c_\alpha}}.
\end{align*}

With the above notation, the \emph{variational derivative} of the Lagrangian $\mathbf L$ in the direction of the vector field of the virtual displacement $\underline\eta$ along the time line $s$ is defined by
\[
\delta\mathbf L(\underline\eta)\vert_{s}:=\delta L_\alpha(\eta_\alpha)\vert_{\gamma_\alpha}.
\]

\begin{thm}
The variation derivative does not depend on the reference frame $\varphi_\alpha$.
\end{thm}

\begin{proof}
Let $\varphi_\beta$ be another reference frame.
Without loss of generality, we can assume $c_\alpha=c_\beta=0$ and $t_\alpha=t_\beta=t$ (we have translation in time between the reference frames $\varphi_\alpha$
and $\varphi_\beta$).
Then we have a situation as in section \ref{s2}, where
\[
Q_\alpha, q_\alpha, \gamma_\alpha, L_\alpha, \eta_\alpha, \quad
Q_\beta, q_\beta, \gamma_\beta, L_\beta, \eta_\beta, \quad \text{and} \quad q_\alpha=g_{\alpha\beta}(q_\beta,t), \eta_\alpha=dg_{\alpha\beta}(\eta_\beta)
\]
corresponds to
\[
Q, q, \gamma, L, \eta, \quad
M, x, \Gamma, l, \xi,\quad \text{and} \quad q=g_t(x), \eta=dg_t(\xi)
\]
respectively. Therefore, the statement follows from Proposition \ref{glavna}.
\end{proof}

We thus have an invariant formulation of classical Lagrangian dynamics on the space-time $\mathcal Q$ in the form of the \emph{d'Alambert principle}:
a time curve $s$ is a motion of the mechanical system defined by the Lagrangian $\mathbf L\colon\mathcal J\to\R$ if the
Lagrangian derivative of $\mathbf L$ is zero,
\[
\delta\mathbf L(\underline\eta)\vert_{s}=0,
\]
for all virtual displacements $\underline\eta$ along $s$.

\begin{thm}
{\sc The Hamiltonian principle of least action} A time line $s(t)$ between the events $\mathbf q_0$ and $\mathbf q_1$
is a motion of the mechanical system defined by the Lagrangian $\mathbf L\colon\mathcal J\to\R$
if and only if it is a critical point of the action integral
\begin{equation*}\label{action***}
S_{\mathbf L}({s})=\int_{\tau(\mathbf q_0)}^{\tau(\mathbf q_1)} \mathbf L(s(t),\dot{s}(t))dt
\end{equation*}
in a class of time lines connecting the events $\mathbf q_0$ and $\mathbf q_1$.
\end{thm}

\subsection*{Acknowledgments}
The author is grateful to Vladimir Dragovi\'c and Borislav Gaji\'c for useful discussions.
The note is based on the joint long year experience in teaching classical mechanics and symplectic geometry at the Mathematical institute SANU.
The author was supported by the Project no. 7744592, MEGIC "Integrability
and Extremal Problems in Mechanics, Geometry and Combinatorics" of the Science Fund
of Serbia.

\end{document}